\documentclass[a4paper]{article}

\usepackage[T1]{fontenc}       
\usepackage{graphicx}          

\usepackage{tikz}              
\usetikzlibrary{shapes,arrows} 
\usepackage[english]{babel}    

\usepackage{amssymb}
\usepackage{amsmath}
\usepackage{amsthm}
\newtheorem{definition}{Definition}
\newtheorem{theorem}{Theorem}

\usepackage{url}
\usepackage{stmaryrd}                   
\usepackage{cite}
\usepackage{xspace}
\usepackage{multicol}

\usepackage{float}    
\usepackage[hidelinks]{hyperref} 
\usepackage{color}

\usepackage{listings}
\newcommand{\inline}[1]{\lstinline{#1}}

\begin{document}

\newcommand\after[1]{\ensuremath{\circ #1}}
\newcommand\doi[1]{\url{https://doi.org/#1}}

\newcommand{\forceindent}{\leavevmode{\parindent=1em\indent}}

\def\Name{\ensuremath{\text{\bf Name}}}
\def\OR{\ensuremath{\ |\ }}
\def\TO{\ensuremath{\rightarrow}}
\def\FROM{\ensuremath{\leftarrow}}
\def\LB{\ensuremath{\llbracket}}
\def\RB{\ensuremath{\rrbracket}}
\def\id{\ensuremath{\mathrm{id}}}

\newcommand\LIT[1]{\ensuremath{\text{\tt #1}}}
\newcommand\SLIT[1]{\ \LIT{#1}\ }
\newcommand\IF[3]{\LIT{if}\ #1\ \LIT{then}\ #2\ \LIT{else}\ #3}
\newcommand\APP[2]{#1\ #2}
\newcommand\CASE[3]{\LIT{case}\ #1\ \LIT{of}\ #2 \TO #3}\newcommand\INBR[1]{\ensuremath{\llbracket #1 \rrbracket}}
\newcommand\MGU[2]{\LIT{unify}(#1, #2)}
\newcommand\INVERT[1]{(\LIT{invert}\ #1)}

\newcommand\EvalWith[4]{\ensuremath{ \Delta#1 #2 \vdash #3 \downarrow #4}}
\newcommand\Eval[2]{\EvalWith{}{\Gamma}{#1}{#2}}
\newcommand\LaveWith[4]{\ensuremath{ \Delta#1 #2 \vdash #3 \uparrow #4}}
\newcommand\Lave[2]{\LaveWith{}{\Gamma}{#1}{#2}}

\newcommand\Because[4]{\ensuremath{\Delta{#1}\INBR{#2 \downarrow #3}\rightsquigarrow #4}}
\newcommand\Esuaceb[4]{\ensuremath{\Delta{#1}\INBR{#2 \uparrow #3}\leftsquigarrow #4}}

\newcommand\LinearHasType[4]{\ensuremath{\Delta#1 #2 \vdash #3 : #4}}
\newcommand\LinearEpytSah[4]{\ensuremath{\Delta#1 #2 \vDash #3 : #4}}
\newcommand\LinearBindTypes[4]{\ensuremath{\Delta#1 | #2 : #3 \Downarrow #4}}
\newcommand\LinearUnBindTypes[4]{\ensuremath{\Delta#1 | #2 : #3 \Uparrow #4}}

\newcommand\UPARROW[1]{$\uparrow$#1}
\newcommand\DOWNARROW[1]{$\downarrow$#1}
\newcommand\RIGHTARROW[1]{$\rightarrow$#1}
\newcommand\LEFTARROW[1]{$\leftarrow$#1}
\newcommand\TYPE[1]{$\tau$#1}
\newcommand\LinINFER[1]{$\Downarrow$#1}
\newcommand\LinUNFER[1]{$\Uparrow$#1}

\newcommand \BX [1]
  {\scriptsize\framebox{{\raisebox{0pt}[0.7\baselineskip][0.01\baselineskip]{\small #1}}}}

\newcommand\Axiom[2]
                 {\ensuremath{\text{\small #1}:\frac{\displaystyle}
                 {\displaystyle #2}}
                 }
\newcommand\InfOne[3]
                 {\ensuremath{\text{\small #2}:\frac{\displaystyle #1}
                 {\displaystyle #3}}
                 }
\newcommand\InfTwo[4]
                 {\ensuremath{\text{\small #3}:\frac{\displaystyle #1 \quad #2}
                 {\displaystyle #4}}
                 }
\newcommand\InfThree[5]
                 {\ensuremath{\text{\small #4}:
                     \frac{\displaystyle #1 \quad #2 \quad #3}
                          {\displaystyle #5}}
                 }
\newcommand\InfFour[6]
                 {\ensuremath{\text{\small #5}:
                     \frac{\displaystyle #1 \quad #2 \quad #3 \quad #4}
                          {\displaystyle #6}}
                 }
\newcommand\InfFive[7]
                 {\ensuremath{\text{\small #6}:
                     \frac{\displaystyle #1 \quad #2 \quad #3 \quad #4 \quad #5}
                          {\displaystyle #7}}
                 }

\def\jeopardy{Jeopardy\xspace}

\title{Tail recursion transformation for invertible functions}

\author{Joachim Tilsted Kristensen$^1$, Robin Kaarsgaard$^2$, \\and Michael Kirkedal Thomsen$^{1,3}$}

\date{$^1$ University of Oslo, Norway\\
$^2$ University of Edinburgh, UK\\
$^3$ University of Copenhagen, Denmark}

\maketitle 

\begin{abstract}
  Tail recursive functions allow for a wider range of optimisations
  than general recursive functions. For this reason, much research has
  gone into the transformation and optimisation of this family of
  functions, in particular those written in continuation passing style
  (CPS).

  Though the CPS transformation, capable of transforming any recursive
  function to an equivalent tail recursive one, is deeply problematic in the
  context of reversible programming (as it relies on troublesome features
  such as higher-order functions), we argue that relaxing (local)
  reversibility to (global) invertibility drastically improves the
  situation.  On this basis, we present an algorithm for tail recursion
  conversion specifically for invertible functions. The key insight is that
  functions introduced by program transformations that preserve
  invertibility, need only be invertible in the context in which the
  functions subject of transformation calls them.  We show how a bespoke
  data type, corresponding to such a context, can be used to transform
  invertible recursive functions into a pair of tail recursive function
  acting on this context, in a way where calls are highlighted, and from
  which a tail recursive inverse can be straightforwardly extracted.
\end{abstract}

\section{Introduction}
\label{sec:introduction}

When a function calls itself, either directly or indirectly, we say
that the function is recursive. Furthermore, when the last operation
of all branches in the definition of a recursive function is the
recursive call, we say that the function is tail recursive.  Unlike
generally recursive functions, tail recursive functions can be easily
compiled into loops in imperative languages (in particular assembly
languages) doing away with the overhead of function calls
entirely. This makes tail recursion a desirable programming style.

Recall that a program is \emph{reversible} when it is written such
that it only consists of invertible combinations of invertible atomic
operations; this is the idea of reversibility as \emph{local}
phenomenon. While every reversible program is also \emph{invertible}
(in the sense that it has an inverse), the converse is not the case,
as an invertible program may consist of a number of non-invertible
functions that simply happen to interact in a way as to make the
program invertible. As such, invertibility is a \emph{global}
phenomenon.

While recursion has been employed in both imperative and functional
reversible programming
languages~\cite{LutzDerby:1986,YokoyamaAxelsenGlueck:2012:LNCS} for
many years, tail recursion has been more cumbersome to handle. Here,
we argue that relaxing (local) reversibility to (global) invertibility
can drastically simplify the handling of tail recursion and even make
it possible to use (adaptations of) conventional CPS transformation
methods for transforming general recursive to tail recursive
functions. To see this, consider the list reversal and list append
functions
\begin{multicols}{2}
\begin{lstlisting}[language=haskell]
reverse1 [      ] = []
reverse1 (x : xs) =
  let ys      = reverse1 xs   in
  let (zs:_x) = snoc1 (ys, x) in
  (zs:_x)
\end{lstlisting}
\begin{lstlisting}[language=haskell]
snoc1 ([    ], x) = x : []
snoc1 (y : ys, x) =
  let (zs:_x) = snoc1 (ys, x) in
  (y : zs:_x)
\end{lstlisting}
\end{multicols}\noindent
The careful reader will have already realised that \inline{reverse1}
is its own inverse. Here, we will refrain from clever realisations and
focus on purely mechanical ways of providing inverse functions. For
instance, the inverses
\begin{multicols}{2}
\begin{lstlisting}[language=haskell]
unsnoc1 (y : zs:_x) =
  let (ys, x) = unsnoc1 (zs:_x) in
  (y : ys, x)
unsnoc1 (x : [   ]) = ([ ], x)
\end{lstlisting}
\begin{lstlisting}[language=haskell]
unreverse1 (zs:_x) =
  let (ys, x) = unsnoc1 (zs:_x) in
  let xs      = unreverse1 ys   in
  (x : xs)
unreverse1 [     ] = [ ]
\end{lstlisting}
\end{multicols}\noindent
are produced by rewriting ``\inline{let y = f x in t}'' to ``\inline{let x =
  unf y in t}'', and then swapping the order bindings in the remaining
program \inline{t}, starting from the last line and ending with the first,
much in the style of Romanenko~\cite{Romanenko:1988}.  To transform these
recursive functions into tail recursive functions, the standard technique is
to introduce an iterator that passes around an explicit argument for
accumulating the deferred part of the computation, e.g.,
\begin{lstlisting}[language=haskell]
reverse2 xs = reverse2_iter (xs, [])

reverse2_iter ([    ], accum) = accum
reverse2_iter (x : xs, accum) = reverse2_iter (xs, x : accum)
\end{lstlisting}
Implementing list reversal in this style makes it tail recursive, but
it also loses an important property, namely \emph{branching
symmetry}. This is crucial, since branching symmetry was the entire
reason why we could mechanically invert the implementations of
\inline{snoc1} and \inline{reverse1}
so easily: because the leaves of their cases are syntactically
orthogonal. For instance, in \inline{reverse1}, when the input is an
empty list, the result is also an empty list, and when the input is
nonempty, the result is also nonempty.

As a consequence of this loss of symmetry, the iterator function
\inline{reverse2\_iter} it is not considered \emph{well-formed for inversion} as defined by Gl\"uck \& Kawabe~\cite{GluckKawabe2004LRParsing}. Consequently, it cannot be implemented in a reversible
functional programming language such as
RFun~\cite{YokoyamaAxelsenGlueck:2012:LNCS,ThomsenAxelsen:2016:IFL} or
$\mathsf{CoreFun}$~\cite{JacobsenEtal:2018}, as it breaks the
symmetric first match policy; the base case returning \inline{accum}
will also return the same value from the iterative case. Even worse,
\inline{reverse2\_iter} cannot be inverted to a deterministic function
using known
methods~\cite{GlueckKawabe:2003,Mogensen:2008,NishidaVidal:2011}. Of
course, this is because \inline{reverse2\_iter} is not injective, so
the outputs of a particular input is not unique.

It does not take much effort to show that \inline{reverse1}
and \inline{reverse2} are semantically equivalent. Thus, since the
latter does nothing but call \inline{reverse2\_iter} it is
surprising that we cannot invert it.
A brief analysis of the problem concludes that \inline{reverse2} restricts
itself to a subset of the domain of \inline{reverse2\_iter}, and
since \inline{reverse2} is clearly injective, \inline{reverse2\_iter} as
restricted to this smaller domain must be injective as well.
By further analysis, we realise that the second component of the
arguments to \inline{reverse2\_iter}, as called by \inline{reverse2}, is
static and can be ignored. In this context \inline{reverse2\_iter} is
in one of three configurations: accepting the restricted input,
iterating, or returning an output. By introducing a data type, we can
explicitly restrict \inline{reverse2\_iter} to this smaller domain:
\begin{lstlisting}[language=haskell]
data Configuration a = Input a | Iteration (a, a) | Output a

reverse3 xs = let (Output ys) = iterate (Input xs) in ys

iterate   (Input xs)               = iterate (Iteration (xs, [    ]))
iterate   (Iteration (x : xs, ys)) = iterate (Iteration (xs, x : ys))
iterate   (Iteration ([    ], ys)) = (Output ys)
\end{lstlisting}
Even further, just like \inline{reverse1} this definition can be
mechanically inverted:
\begin{lstlisting}[language=haskell]
uniterate (Output ys)              = uniterate (Iteration ([    ], ys))
uniterate (Iteration (xs, x : ys)) = uniterate (Iteration (x : xs, ys))
uniterate (Iteration (xs, [    ])) = (Input xs)

unreverse3 ys = let (Input xs) = uniterate (Output ys) in xs
\end{lstlisting}
Moreover, these four function definitions are all tail recursive,
which was what we wanted.

\emph{Structure:} In this article we will show an algorithm that can perform
this transformation. First, in Section~\ref{sec:general-idea}, we illustrate
the program transformation by example, before describing it formally in
Section~\ref{sec:formal-transformation} and prove its
correctness. Afterwards, we discuss a couple of known limitations
(Section~\ref{sec:limitations}) of our approach, and show how the resulting
constructs can be compiled to a flow-chart language
(Section~\ref{sec:translation}). Finally, we discuss related in
Section~\ref{sec:related-work} and end in Section~\ref{sec:conclusion} with
some concluding remarks.

\section{Tail recursion transformation, by example}
\label{sec:general-idea}

The transformation we propose assumes a functional programming
language with first order functions, algebraic datatypes, and
recursion, as these are the features commonly found in reversible
functional programming
languages~\cite{YokoyamaAxelsenGlueck:2012:LNCS,ThomsenAxelsen:2016:IFL,JacobsenEtal:2018,JamesSabry:2014:RC}.
Moreover, as the subject of the transformation, we only consider
functions that are \emph{well-formed for
inversion}~\cite{GluckKawabe2004LRParsing} as usual, meaning that the
patterns of case-expressions are orthogonal, either syntactically, or
by guard statements as suggested in Mogensen's semi-inversion for
guarded-equations~\cite{Mogensen:2005}. Furthermore, we require that
expressions and patterns are linear (any variable binding is used
exactly once), and (for simplicity) that a variable cannot be
redefined in expressions that nest binders (such as \inline{let} and
\inline{case}).

Such programming languages usually introduce the notion of tail recursion
by introducing an imperative style language feature. For instance, Mogensen's
language for guarded equations~\cite{Mogensen:2005} features a \inline{loop}
construct that allows it to call a partial function until it fails (by
pattern matching not exhaustive), as illustrated by the function
\inline{reverse4}, defined by:
\begin{lstlisting}[language=haskell]
reverse4(xs) = let ([], acc) = loop revStep (xs, []) in acc
  where revStep(x : xs, acc) = (xs, x : acc);
\end{lstlisting}
\noindent
Likewise, the Theseus programming language~\cite{JamesSabry:2014:RC}
provides a trace operation encoded via so-called \emph{iteration labels}, as demonstrated in \inline{reverse5} below.
\begin{lstlisting}[language=haskell]
iso reverse5 :: [a] <-> [a]
| xs                        = iterate $ inL xs, []
| iterate inL $ (x : xs) ys = iterate $ inL xs, (x : ys)
| iterate inL $ [],      ys = iterate $ inR ys
| iterate     $ inR      ys = ys
  where
    iterate :: ([a] * [a]) + [a]
\end{lstlisting}

\noindent
We \emph{do not} introduce a new language feature, but instead relax
the requirement that \emph{all functions} must be well-formed
for inversion. Instead we require only that \emph{the subject of the
transformation} must be well-formed for inversion. For instance,
recall that the function \texttt{snoc1} from
Section~\ref{sec:introduction} is well-formed for inversion, and
consider Nishida \& Vidal's CPS transformation of first-order
functions~\cite{NishidaVidal:2014}
\begin{lstlisting}[language=haskell]
data Continuation a = Id | F a (Continuation a)

snoc2 p = snoc2_iter (p, Id)
  where
    snoc2_iter (([    ], x), g) = snoc2_call (g, x : [])
    snoc2_iter ((y : ys, x), g) = snoc2_iter ((ys, x), F y g)
    snoc2_call (Id   , zs:_x) = zs:_x
    snoc2_call (F y g, zs:_x) = snoc2_call (g, y : zs:_x)
\end{lstlisting}
\noindent
Here, the computation has been split into two parts; one that computes
a structure corresponding to the closure of the usual continuation
function, and another that corresponds to evaluating the call to said
continuation.  Now, just as with \inline{reverse2\_iter},
\inline{snoc2\_iter} and \inline{snoc2\_call} are not injective
functions, but can be restricted to such when recognizing that one or
more of its arguments are static (\inline{Id} and \inline{[]}
respectively).
Consequently, we can introduce a datatype that does away with these,
and invert \inline{snoc2} as
\begin{lstlisting}[language=haskell]
data Configuration' input acc arg output =
    Input'    input
  | Iterate  (input, Continuation acc)
  | Call     (Continuation acc, arg)
  | Output'   output

snoc6 (ys, x) =
  let (Output' (zs:_x)) = snoc6_call (snoc6_iter (Input' (ys, x))) in  (zs:_x)

snoc6_iter (Input' (ys, x))                = snoc6_iter (Iterate ((ys, x), Id))
snoc6_iter (Iterate ((y : ys, x), g))      = snoc6_iter (Iterate ((ys, x), F y g))
snoc6_iter (Iterate (([    ], x), g))      = (Call (g, [x]))

snoc6_call (Call (g, [x]))                 = snoc6_call (Iterate ([x], g))
snoc6_call (Iterate ((zs:_x), F y g))      = snoc6_call (Iterate (y : (zs:_x), g))
snoc6_call (Iterate ((zs:_x), Id))         = (Output' ((zs:_x)))

unsnoc6_call (Output' ((zs:_x)))           = unsnoc6_call (Iterate ((zs:_x), Id))
unsnoc6_call (Iterate (y : (zs:_x), g))    = unsnoc6_call (Iterate ((zs:_x), F y g))
unsnoc6_call (Iterate ([x], g))            = (Call (g, [x]))

unsnoc6_iter (Call (g, [x]))               = unsnoc6_iter (Iterate (([ ], x), g))
unsnoc6_iter (Iterate ((ys, x), F y g))    = unsnoc6_iter (Iterate ((y : ys, x), g))
unsnoc6_iter (Iterate ((ys, x), Id))       = (Input' (ys, x))

unsnoc6 (zs:_x) =
  let (Input' (ys, x)) = unsnoc6_iter (unsnoc6_call (Output' (zs:_x))) in (ys, x)
\end{lstlisting}
Moreover, because the iterator does not use \inline{Output'} and the call
simulation does not use \inline{Input'}, we can introduce two separate
datatypes, and a couple of gluing functions to improve composition.
\begin{lstlisting}[language=haskell]
data Configuration2 input acc arg
  = Input2   input
  | Iterate2 (input, Continuation acc)
  | Output2  (Continuation acc, arg)

data Configuration3 acc arg output
  = Input3   (Continuation acc, arg)
  | Iterate3 (Continuation acc, output)
  | Output3  output

input      a           = (Input2 a)
uninput   (Input2  a)  = a
glue      (Output2 a)  = (Input3 a)
unglue    (Input3   a) = (Output2 a)
output     a           = (Output3 a)
unoutput  (Output3 a)  = a
\end{lstlisting}
Now, because \inline{reverse1} was also well-formed for inversion, we
can apply the usual CPS transformation, and obtain a tail-recursive inverse
program by the exact same procedure:
\begin{lstlisting}[language=haskell]
reverse6   = unoutput . call6 . glue . iterate6 . input

iterate6   (Input2 xs)               = iterate6   (Iterate2 (xs, Id))
iterate6   (Iterate2 (x : xs, g))    = iterate6   (Iterate2 (xs, F x g))
iterate6   (Iterate2 ([], g))        =            (Output2  (g, []))
call6      (Input3  (g, []))         = call6      (Iterate3 (g, []))
call6      (Iterate3 (F x g, ys))    = call6      (Iterate3 (g, zs:_x))
  where zs:_x = snoc6 (ys, x)
call6      (Iterate3 (Id  , ys))     =            (Output3 ys)
uncall6    (Output3 ys)              = uncall6    (Iterate3 (Id  , ys))
uncall6    (Iterate3 (g, zs:_x))     = uncall6    (Iterate3 (F x g, ys))
  where (ys, x) = unsnoc6 (zs:_x)
uncall6    (Iterate3 (g, []))        =            (Input3  (g, []))
uniterate6 (Output2  (g, []))        = uniterate6 (Iterate2 ([], g))
uniterate6 (Iterate2 (xs, F x g))    = uniterate6 (Iterate2 (x : xs, g))
uniterate6 (Iterate2 (xs, Id))       =            (Input2 xs)

unreverse6 = uninput . uniterate6 . unglue . uncall6 . output
\end{lstlisting}

It is important to note that even though \inline{reverse6} seems a bit more
complicated than \inline{reverse3}, we did not start with a tail recursive
function, and the transformation process was entirely mechanical. First we
converted into tail recursive form using continuation passing style. Then,
we restricted the functions introduced by the transformation, to the domain
on which they are called by the function you are inverting. Finally, we
inverted all the operations performed by \inline{reverse1}, and this was
also entirely mechanical (since \inline{reverse1} was well-formed for
inversion), and produced the inverse program by swapping the input and
output arguments (keeping the recursive call in front of the
\inline{Iterate} data structure).

\section{Tail recursion transformation, formally}
\label{sec:formal-transformation}

In the interest of simplicity we will show how the transformation
works on a small, idealised subset of the Haskell programming language
as shown in Figure~\ref{fig:syntax}, restricted to first order
function application and conditionals.
\begin{figure}[tp]
\begin{center}
  \begin{align*}
    p      &::= c(p_i)                                       &\text{(Constructor).}\\
           &\OR x                                            &\text{(Variable).}\\
    t      &::= p \OR \APP{f}{p} \OR \CASE{t}{p_i}{t_i}      &\text{(Terms.)}\\
\Delta     &::= \APP{f}{p} =\ t\ [\texttt{where}\ p_i = t_i].\ \Delta
\OR \LIT{data}\ \tau\ =\ c_j(\tau_i)\ .\ \Delta \OR \epsilon &\text{(Programs).}
\end{align*}
\end{center}
\caption{The syntax for a first order functional programming language.}
\label{fig:syntax}
\end{figure}
A term is a pattern, a function applied to a pattern, or a
case-expression, though in program examples, we might use a
\inline{where}-clause or a \inline{let}-statement when the syntactic
disambiguation is obvious. Functions applied to terms, patterns that
consist of terms, and let-statements are disambiguated as shown in
Figure~\ref{fig:sugar}.

\begin{figure}[tp]
\begin{align*}
  \LB \APP{f}{t} \RB &::= \CASE{\LB t \RB}{p}{\APP{f}{p}}\\
  \LB c(t_i) \RB &::= \CASE{\LB t_i \RB}{p_i}{c(p_i)}\\
  \LB \LIT{let}\ p\ =\ t_1\ \LIT{in}\ t_2 \RB &::= \CASE{\LB t_1 \RB}{p}{\LB t_2 \RB}\\
  \LB \APP{f}{p_i} = t_i \RB &::= \APP{f}{x} =\ \CASE{x}{p_i}{\LB t_i \RB}\\
  \LB \APP{f}{p} = t\ \LIT{where}\ p_i = t_i \RB &::= \APP{f}{p} = \CASE{\LB t_i \RB}{p_i}{ \LB t \RB}
\end{align*}
\caption{Disambiguation of syntactic sugar.}
\label{fig:sugar}
\end{figure}

First, we give the definition of the requirements for the
transformation to work.

\begin{definition}\label{def:closed}
A term $t$ is closed under a pattern $p$ precisely if all of the variables
that occur in $p$ appear in $t$ exactly once.
\end{definition}

\begin{definition}\label{def:well-formed}
A function $f$, as defined by the equation $\APP{f}{p} = t$, is \emph{well-formed for inversion}, if $t$ is closed under $p$. Moreover,
\begin{itemize}
  \item If $t$ is an application, then $t \equiv \APP{g}{p_0}$, where $g$ is
    well-formed for inversion as well.
  \item If $t$ is a case-expression, then $t \equiv \CASE{t_0}{p_i}{t_i}$, where
    then $t_0$ is well-formed for inversion, each $t_i$ is closed under the
    corresponding pattern $p_i$, and for all indices $j$ and $k$, if $j < k$
    then $p_j$ is syntactically distinguishable from $p_k$ and the leaf
    terms of $t_j$ are all syntactically distinguishable from the
    corresponding leaf terms of $t_k$.
\end{itemize}
\end{definition}

\noindent
When a function is \emph{well-formed for inversion} in this way, we know how
to invert it using existing methods, even though such methods may require some
expensive search. However, functions that do not contain a
\LIT{case}-expression are all trivially and efficiently invertible, and we
can focus on the hard part, namely conditionals.

Functions that are well-formed for inversion will be implemented with function clauses of the following two forms
\begin{align*}
  \APP{f}{p_k} &= t_k\\
  \APP{f}{p_i} &= \APP{g_i}{(t_{i0}, t_{i1})}.
\end{align*}
\noindent
Each term $t_k$ is well-formed for inversion and do not contain recursive
calls to $f$, $k$ is less than $i$, and $g_i$ is well-formed for inversion.
Furthermore, $t_{i0}$ may contain recursive calls to $f$ but $t_{i1}$ is free of
such calls. Moreover, the result of calling $g_i$ with these arguments
yield patterns that are distinguishable from the results of calling $g_j$ on
$(t_{j0}, t_{j1})$ whenever $i < j$.

The first order CPS transformation proposed by Nishida \& Vidal
essentially defers the call to $g_i$ by storing the unused parts of
$p_i$ and $t_{i1}$ in a data structure, yielding the program transforms
\begin{align*}
  \LIT{data}\ \tau    &= \LIT{Id}\ \OR \LIT{G}_i(\tau_{t_{i1}}, \tau)\\
  \APP{f_0}{x}        &= \APP{f_1}{(\LIT{Id}, x)}\\
  \APP{f_j}{(g, p_k)} &= \APP{f_n}{(g, t_k)}\\
  \APP{f_j}{(g, p_i)} &= \APP{f_l}{(\LIT{G}_i(t_{i1}, g), t_{i0})}\\
  \APP{f_{n}}{(\LIT{Id}, y)} &= y\\
  \APP{f_{n}}{(\LIT{G}_i(p_{i1}, g), p_{i0})} &= \APP{f_{n}}{(g, \APP{g_i}{(p_{i0}, p_{i1})})} \, ,
\end{align*}
\noindent
where $1 \leq j \leq n$ and $1 \leq l \leq n$.
This transformation clearly preserves semantics (in the sense that $f$
is semantically equivalent with $f_0$) since $f_j$ essentially
builds up a stack of calls to respective $g_i$'s, while $f_n$ performs
these calls in the expected order.

The only problem is that each $f_j$ may not be well-formed for
inversion, and $f_n$ is certainly not well-formed; the variable
pattern in the first case is a catch-all that later cases cannot be
syntactically orthogonal to.
Consequently, we cannot use existing methods to invert these
functions. Instead, we realize that their origins are well-formed for
inversion, so we should have been able to invert them in a way that is
``well formed enough''. The idea is to represent each intermediate
function with a datatype, and use the fact that each $g_i$ is
well-formed for inversion to construct the invertible program as
\begin{align*}
  \LIT{data}\ \tau               &= \LIT{Id}\ \OR \LIT{G}_i(\tau_{t_{i1}}, \tau).\\
  \LIT{data}\ \tau'              &= \LIT{In}(\tau_x) \OR \LIT{F}_{j0}(\tau, \tau_{i0}) \OR \LIT{H}_{k1}(\tau, \tau_k).\\
  \LIT{data}\ \tau''             &= \LIT{H}_{k2}(\tau, \tau_k) \OR \LIT{Eval}(\tau, \tau_k) \OR \LIT{Out}(\tau_y).\\
  f'_0                           &= f'_2 \circ h \circ f'_1.\\
  \APP{h}{(\LIT{H}_{k1}(f, x))}   &= \LIT{H}_{k2}(f, x).\\
  \APP{f'_1}{(\LIT{In}(p_l))}       &= \APP{f'_1}{(\LIT{F}_l(\LIT{Id}, p_l))}\\
  \APP{f'_1}{(\LIT{F}_j(g, p_j))}   &= \APP{f'_1}{\LIT{F}_{j0}(\LIT{G}_j(t_{i1}, g), t_{i0})}.\\
  \APP{f'_1}{(\LIT{F}_j(g, p_k))}   &= \APP{f'_1}{\LIT{H}_{k1}(g, p_k)}.\\
  \APP{f'_2}{(\LIT{H}_{k2}(g, p_k))} &= \APP{f'_2}{\LIT{Eval}(g, p_k)}.\\
  \APP{f'_2}{(\LIT{Eval}(\LIT{F}_{i0}(\LIT{G}_i(p_{i1}, g)), p_{i0}))} &= \APP{f'_2}{(g, y_i)}\
  \LIT{where}\ y_i = g_i(p_{i0}, p_{i1}).\\
  \APP{f'_2}{(\LIT{Eval}(\LIT{Id}, y))} &= \LIT{Out}(y).
\end{align*}
Now, just as with the CPS transformation, $f'_0$ is semantically equivalent
to $f$ because $f'_1$ collects calls $\LIT{G}_j$ and $f'_2$ evaluates
them. As such, the only difference is that the input is wrapped in \LIT{In}
and \LIT{Out}. However, this time we can derive an inverse program as

\begin{align}
  f'^{-1}_0                           &= f'^{-1}_1 \circ h^{-1} \circ f'^{-1}_2\\
  \APP{h^{-1}}{(\LIT{H}_{k2}(f, x))}   &= \LIT{H}_{k1}(f, x)\\
  \APP{f'^{-1}_2}{(\LIT{Out}(y))} &= \APP{f'^{-1}_2}{(\LIT{Eval}(\LIT{Id}, y))}\\
  \APP{f'^{-1}_2}{(g, y_i)} &=  \APP{f'^{-1}_2}{(\LIT{Eval}(\LIT{F}_{i0}(\LIT{G}_i(p_{i1}, g)), p_{i0}))}\\
  & \LIT{where}\ (p_{i0}, p_{i1}) = g^{-1}_i(y_i)\\
  \APP{f'^{-1}_2}{(\LIT{Eval}(g, p_k))} &= \LIT{H}_{k2}(g, p_k)\\
  \APP{f'^{-1}_1}{(\LIT{H}_{k1}(g, p_k))}   &= \APP{f'^{-1}_1}{(\LIT{F}_j(g, p_k))}\\
  \APP{f'^{-1}_1}{\LIT{F}_{j0}(\LIT{G}_j(p_{i1}, g), p_{i0})} &= \APP{f'^{-1}_1}{(\LIT{F}_j(g, p_j))}\\
  \APP{f'^{-1}_1}{(\LIT{F}_l(\LIT{Id}, p_l))}  &=  (\LIT{In}(p_l))
\end{align}

The correctness of this technique can be shown as follows.

\begin{theorem}
  The function $f'^{-1}_0$ is inverse to $f_0$.
\end{theorem}
\begin{proof}
We remark the following for each step of the transformation:
\begin{itemize}
\item (1) $f'^{-1}_0$ is inverse to $f'_0$ (by definition of function
  composition) precisely if $h^{-1}$, $f'^{-1}_1$, and $f'^{-1}_2$ are the
  inverse to $h$, $f'_1$, and $f'_2$ respectively.
\item (2) $h^{-1}$ is trivially inverse to $h$ since it does not contain
  application or case-expressions.
\item (3) There is only one way of constructing the arguements to
  $f'^{-1}_2$, namely using the constructor $\LIT{Out}$ on the output of $f'_2$.
\item (4-5) Since $g_i$ was well-formed for inversion, the output $y_i$ is
  syntactically orthogonal to outputs of $g_j$ when $i \neq j$. The patterns
  it takes as arguments $(p_{i0}, p_{i1})$ are syntactically orthogonal to
  all other such patterns, so the choice of constructors $\LIT{F}_{i0}$ and
  $\LIT{G}_i$ has to be unique as well.
\item (6) $p_k$ is trivially recognized as one of the syntactically
  orthogonal parts of the left-hand side of $f$, which was well-formed
  for inversion.
\item (7) There is only one way of constructing $\LIT{H}_{k1}$, namely
  using $h^{-1}$.
\item (8) These are exactly the arguments of $g_j$, Since $f$ was
  well-formed for inversion, they must be closed under $p_j$
  (Definition~\ref{def:well-formed}), which we may now reconstruct by
  copying.
\item (9) Finally, the first argument of $\LIT{F}_l$ could only have
  been $\LIT{Id}$ in one program point, and the result has to be
  constructed using $\LIT{In}$, and we are done.
\end{itemize}
By equations (7)--(9), $f'^{-1}_1$ is inverse to $f'_1$, and by
equations (3)--(6), $f'^{-1}_2$ is inverse to $f'_2$. Since, by
equation (2), $h^{-1}$ is inverse to $h$, it follows by equation (1)
that $f'^{-1}_0$ is inverse to $f'_0$. Now, since $f'_0$ was
semantically equivalent to $f_0$, $f'^{-1}_0$ must be inverse to $f_0$ as
well, and we are done.
\end{proof}

\section{Known limitations}
\label{sec:limitations}
In Definition~\ref{def:closed} we required linearity, which is
slightly stronger than it needs to be. The reason why we chose this
restriction is because it commonly occurs in reversible
programming~\cite{YokoyamaAxelsenGlueck:2012:LNCS,ThomsenAxelsen:2016:IFL},
and makes it easy to reject programs that are trivially
non-invertible. However, the linearity restriction could be relaxed to
\emph{relevance} (i.e., that variables must occur \emph{at least once}
rather than \emph{exactly once}) as in
\cite{JacobsenEtal:2018}. Moreover, we might even want to relax this
restriction even further to say that all values that were available to
a particular function of interest must be used at least once on every
execution path. We do not believe that it can relaxed further than
that, as an invertible program cannot lose information when it is not
redundant.

Additionally, one may want to relax the constraints of local invertibility
to be operations for which an inverse is symbolically derivable. For
instance, consider extending the syntax for patterns with integer literals,
and terms with addition and subtraction.  Hence, the following formulation of
the Fibonacci-pair function is possible.
\begin{lstlisting}[language=haskell]
fib   (a, b) = (a + b, a)
dec   n      = n - 1

fib_pair 0 = (1, 1)
fib_pair n = fib (fib_pair (dec n))
\end{lstlisting}
While this program is invertible, it requires a bit of inference to
derive the inverse. For instance, that one of the arguments of
\inline{fib} is preserved in its output, which is needed to infer
\inline{unfib}. Likewise for \inline{dec} and \inline{undec}, the
compiler must infer that subtracting a constant can be automatically
inverted.

Additionally, while the algebraic data-representation of natural
number constants is syntactically distinguishable, with integer
constants and variables the compiler has to insert guards, as in
\begin{lstlisting}[language=haskell]
fib_pair =  unoutput . call7 . glue . iterate7 . input

iterate7    (Input2 n)                      = iterate7    (Iterate2 (n, Id))
iterate7    (Iterate2 (n, f)) | n /= 0      = iterate7    (Iterate2 (n', F () f))
  where n' = dec n
iterate7    (Iterate2 (0, f))               =             (Output2 (f, (1,1)))
call7       (Input3   (f, (1, 1)))          = call7       (Iterate3 (f, (1, 1)))
call7       (Iterate3 (F () f, pair))       = call7       (Iterate3 (f, y))
  where y = fib pair
call7       (Iterate3 (Id , x))             =             (Output3 x)
uncall7     (Output3 x)                     = uncall7     (Iterate3 (Id, x))
uncall7     (Iterate3 (f, y)) | y /= (1, 1) = uncall7     (Iterate3 (F () f, pair))
  where pair = unfib y
uncall7     (Iterate3 (f, (1, 1)))          =             (Input3 (f, (1, 1)))
uniterate7  (Output2  (f, (1, 1)))          = uniterate7  (Iterate2 (0, f))
uniterate7  (Iterate2 (n', F () f))         = uniterate7  (Iterate2 (n, f))
  where n = undec n'
uniterate7  (Iterate2 (n, Id ))             =             (Input2 n)

unfib_pair = uninput . uniterate7 . unglue . uncall7 . output

unfib (ab, a) = (a, ab - a)
undec      n  = n + 1
\end{lstlisting}
However, the necessary guards are essentially predicates stating that
future clauses do not match (so, they can all be formulated using the
$\neq$-operator). Moreover, the additional meta theory needed for this
kind of support is fairly simple. In this case, that adding a constant
can be inverted by subtracting it, and that one of the arguments of
addition must be an available expression in the term returned by the
call.

\section{Translation to flowchart languages}
\label{sec:translation}
One reason for putting recursive functions on a tail recursive form is
for efficiency, as tail recursive programs can be easily compiled to
iterative loop-constructs, eliminating the overhead of function
calls. We sketch here how the transformed programs can be translated
to a reversible loop-construct from flowchart
languages~\cite{YokoyamaAxelsenGlueck:2008:ICALP} (see also
\cite{GlueckKaarsgaard:2018:LMCS, GlueckEtAl:2022:TCS}), which
can be implemented in
Janus~\cite{LutzDerby:1986,YokoyamaGlueck:2007:Janus} and later be
compiled~\cite{Axelsen:2011:CC} to reversible abstract machines such
as PISA~\cite{Vieri:1999} or
BobISA~\cite{ThomsenAxelsenGluck:2012:LNCS}.

We remind the reader that the reversible loop has the following structure:
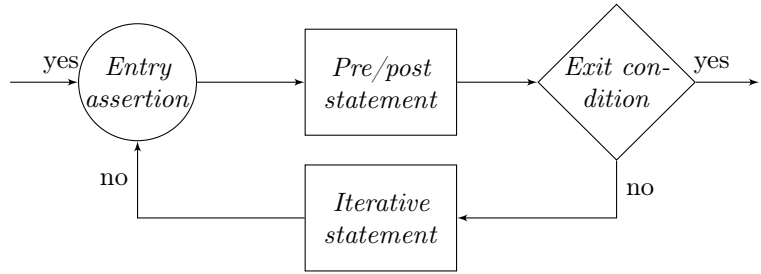
\begin{center}
\tikzstyle{assert} = [circle, draw, text width=4em, text badly centered, node distance=3cm, inner sep=0pt, minimum height=4em]
\tikzstyle{condition} = [diamond, draw, text width=4em, text badly centered, node distance=3cm, inner sep=0pt, minimum height=4em]
\tikzstyle{statement} = [rectangle, draw, text width=5em, text centered, minimum height=4em]
\tikzstyle{line} = [draw, -latex']

\begin{tikzpicture}[node distance = 1cm, auto]
    \node [] (start) {};
    \node [assert, right of=start,node distance=1.8cm] (as) {\emph{Entry assertion}};
    \node [right of=as,node distance=1.6cm] (mida) {};
    \node [statement, right of=mida,node distance=1.6cm] (stmtp) {\emph{Pre/post statement}};
    \node [right of=stmtp,node distance=1.6cm] (midc) {};
    \node [condition, right of=midc,node distance=1.5cm] (cond) {\emph{Exit condition}};
    \node [right of=cond,node distance=2cm] (final) {};
    \node [statement, below of=stmtp,node distance=1.8cm] (stmti) {\emph{Iterative statement}};

    \path [line] (start) -- node [near end] {yes} (as) ;
    \path [line] (as) --  (stmtp) ;
    \path [line] (stmtp) --  (cond) ;
    \path [line] (cond) |- node [near start] {no} (stmti) ;
    \path [line] (cond) -- node [near start] {yes} (final) ;
    \path [line] (stmti) -| node [near end] {no} (as) ;
\end{tikzpicture}
\end{center}
The entry assertion must only be true on entry to the loop, while the
exit condition will only be true in the final iterations. For
completeness there are two statements in the loop: the upper (called \emph{pre/post statement}) we can use to transform between the input/output state and the iterative state, while the lower (called
\emph{iterative statement}) is the most widely used as this has
similar semantics to the normal while-loop.

We will show the translation based on the \inline{reverse3} example
from before.
\begin{lstlisting}[language=haskell]
data Configuration a = Input a | Iteration (a, a) | Output a

reverse3 xs = let (Output ys) = interpret (Input xs) in ys

interpret   (Input xs)               = interpret (Iteration (xs, [    ]))
interpret   (Iteration (x : xs, ys)) = interpret (Iteration (xs, x : ys))
interpret   (Iteration ([    ], ys)) = (Output ys)
\end{lstlisting}
The first step is to apply our transformation to yield a tail recursive function; here, this has already been done. 
Next, we must translate the functional abstract data types to imperative values. 
The \inline{Configuration} type will be translated into an enumeration type, with the values \inline{Input}, \inline{Iteration}, and \inline{Output} encoded at integers (e.g. 1, 2, and 3). We would also need to encoded the function data (here the two lists), which could be done with an arrays and a given length. We will, however, not dwell on the data encoding, as our focus is the translation of code that our translation generates.

We can now construct the \inline{reverse3} procedure that will contain the loop. This will be given the encoded list and return the reversed encoded list. Here a full compiler (again outside our scope) should also be aware that e.g. Janus restricts to call-by-reference, making it needed to compile the function to inline data handling. Though, this is not a restriction in reversible assembly languages.
In the beginning of \inline{reverse3} we will create a local variable \inline{configuration} that is initialised to \inline{Input}. After the loop, this variable will be delocalised with the value \inline{Output}.
At the entry to the loop, the available variables will, thus, be \inline{configuration} and the function data (i.e. the encoding of the two lists).

The reversible loop will implement the \inline{interpret} function. We assume that there exist a translation of the data handling, meaning that we have the two procedures
\begin{description}
    \item[\texttt{empty}] that checks if the encoded list is empty, and
    \item[\texttt{move}] that move the first element of an encode list to the other.
\end{description}
With this, we mechanically derive the four components of the loop as
\begin{description}
    \item[Entry assertion: \texttt{configuration = Input}.] We have defined that is the only valid value at entrance. Afterwards it will not be used.
    \item[Exit condition: \texttt{configuration = Output}.] Similar to before, this value is only used on exit from the function.
    \item[Pre/post statement: Line 1 and 3.] These two lines can be implemented as two conditions in sequence, similar to
\begin{lstlisting}{language=c}
  if   (configuration = Input) 
  then configuration++ // Update from enum Input to Iteration
  fi   (configuration = Iteration and empty(ys))

  if   (configuration = Iteration and empty(xs))
  then configuration++ // Update from enum Iteration to Output
  fi   (configuration = Output)
\end{lstlisting}
    Here, the first condition transforms the \inline{Input} value to an \inline{Iteration} value with the assertion that the resulting list is empty, while the second condition transforms a \inline{Iteration} value with an empty list to an \inline{Output} value with an assertion that we now have an output value. 
    
    \item[Iterative statement: Line 2.] This performs the iterative computation, generating code similar to
\begin{lstlisting}{language=c}
  if   (configuration = Iteration and (not empty(xs))) 
  then move(xs,ys) // Update from enum Input to Iteration
  fi   (configuration = Iteration and (not empty(ys)))
\end{lstlisting}
    For completeness we check and assert that \inline{configuration = Iteration}, though this is clear from the translation. We also assure correct data handling, by checking that the relevant lists are non-empty (matching the pattern matching of the function) and implement the relevant data handling (the move function).
\end{description}
The generated program could be more efficient, but it clearly demonstrates
how the datatype \inline{Configuration} translates to a reversible
loop. The hard work is in the encoding of the data.
This approach also applies to functions that have more one function clause with \inline{Input} and
\inline{Output} cases, and more iterative clauses.

\section{Discussion and related work}
\label{sec:related-work}
While it is possible to invert all injective
functions~\cite{McCarthy:1956, AbramovGlueck:2002}, inverse programs
constructed this way are often not very efficient. In spite of this,
specific inversion methods tend to have well-defined subsets of
programs for which they can produce efficient inverses.

Precisely classifying the problems which can be efficiently inverted
is hard, so the problem is usually approached from a program-specific
perspective. One approach is to restricting programs to be formulated
in a way that is particularly conducive to inversion. Another approach
is grammar-based-inversion, which works by classifying how hard it is
to invert a function, based on the properties of a grammar derived
from the function body that decides whether or not a given value is in
its range~\cite{Matsuda2010grammar-based,GlueckKawabe:2005:LR}.

An alternative perspective on finding efficient inverse programs is to
acknowledge the huge body of knowledge that has been produced in order to
optimized programs running in the forward direction for time complexity, and
see if we can bring those optimizations into the realm of reversible
computing.
In doing so we have not found a need to invent new class of programs
to invert. Instead, we enable existing techniques for optimizing CPS
transformed programs to be leveraged on programs which do not
naturally allow for CPS transformation.

The technique we use for transforming programs into tail recursive
form is essentially Nishida \& Vidal's method for continuation passing
style for first order programs\cite{NishidaVidal:2014}. In doing so, we
introduce an extra function that evaluates a data type that represents
a continuation.

In related work on grammar/syntax based inversion
techniques~\cite{GluckKawabe2004LRParsing,NishidaVidal:2011},
\emph{well-formed with respect to inversion} means that the function is
linear in its arguments (and so does not throw anything away), and that
cases are syntactically orthogonal.
Programs that are well-formed in this sense allow inversion by
applying known inversion methods to the iteration function, which then
becomes a non-deterministic inverse program (since it need not be
injective). However, existing methods for non-determinism elimination
can be applied to solve this problem since the original program was
\emph{well-formed}.

\section{Conclusion}
\label{sec:conclusion}

In this work we have shown that invertible programs admit a tail
recursion transformation, provided that they are syntactically
well-formed. This was achieved using a version of the first order
CPS transformation tailored to invertible programs.
Alternatives that do not have tail recursion optimisation must instead
rely on search, which can be prohibitively expensive.
Instead of searching, we can enforce determinism by pattern
matching. That is, transformations where the non-injective part is
introduced by the compiler, we can use a ``\textit{new
  datatype trick}''. Finally, we have shown correctness of our
transformation and how the transformed programs can be efficiently
compiled to the reversible loops found in reversible flowchart
languages, which in turn may serve as a basis for efficient
implementations in reversible abstract machines.

\subsection{Future work}
Currently, the transformation is implemented for at subset of
Haskell. Future work will be to integrate this into a invertible
functional programming languages such as
\jeopardy~\cite{KristensenEtal:2022:IFLwip,KristensenEtal:2022:NIK}.

This work avoids the need for a symbolic and relational intermediate
representation. Perhaps future iterations on such an approach will
enable a relaxation of the existing methods' very strict requirements
(such as linearity), and thus a less restrictive notion of
well-formedness, but also a less syntactic notion of the complexity of
function invertibility.

A major improvement to the complexity of function invertibility would
also be to eschew classifying \emph{programs} that are hard to invert
in favor of classifying \emph{problems}. One approach could be to see
if the grammar-based approach from \cite{Matsuda2010grammar-based} can
be relaxed to grammars that recognize the \emph{output} of the
function, rather than grammars \emph{generated by the syntactic
structure of the output} of a program.

An example of such a relaxation would to allow existential variables. That
is, to split the mechanism of introducing a variable symbol from the
mechanism that associates it with a value (its binder).  This is customary
in logic programming languages such as \inline{Prolog}, where programs
express logical relationships that are solved for all possible solutions
based on backtracking that redefines variable bindings.  In a functional
language, such a mechanism could try to postpone the need to use a free
variable until as late as possible, allowing partially invertible functions
that accept and return partial data structures (containing logical
variables) that may be combined to complete ones (free of logical variables)
when composed in certain ways. We are currently exploring this concept
further in related work on the Jeopardy programming
language~\cite{KristensenEtal:2022:IFLwip}.

The use of existential variables could further enable the relaxation
of the linearity constraint beyond relevance, such that an iterator
function may reconstruct a partial term (containing free varaibles)
which is then unified with the available knowledge about its origin,
if it is possible to unify it to a complete term (not containing free
variables). We have developed an analysis to infer per-program-point
sets of such information\cite{KristensenEtal:2022:NIK}, which may be
combined with control flow analysis to decide on a suitable program
point in which to unify.

\bibliographystyle{splncs04}
\bibliography{sample-base}

\end{document}